\documentclass{llncs}

\usepackage{amsmath,amssymb}
\usepackage{enumitem}
\usepackage[dvips]{epsfig}
\usepackage{epsf}
\usepackage{float}
\usepackage{array}

\def\fmod#1 #2{#1\ ({\rm mod}\ #2)}

\usepackage[dvips]{epsfig}
\usepackage{epsf}

\makeatletter
\def\Ddots{\mathinner{\mkern1mu\raise\p@
\vbox{\kern7\p@\hbox{.}}\mkern2mu
\raise4\p@\hbox{.}\mkern2mu\raise7\p@\hbox{.}\mkern1mu}}
\makeatother

\frontmatter

\title{Automatic Theorem-Proving in Combinatorics on Words}

\author{Dane Henshall \and Jeffrey Shallit}

\institute{School of Computer Science, University of Waterloo,
Waterloo, ON  N2L 3G1 Canada  \\
\email{dhenshall@uwaterloo.ca, shallit@cs.uwaterloo.ca}
}

\begin{document}

\maketitle

\begin{abstract}
We describe a technique for mechanically proving certain kinds of theorems in
combinatorics on words, using automata and a package for manipulating
them.  We illustrate our technique by solving, 
purely mechanically, an open problem of Currie and Saari on the lengths
of unbordered factors in the Thue-Morse sequence.
\end{abstract}

\centerline{\it Dedicated to the memory of Sheng Yu (1950--2012):  friend and colleague}

\section{Introduction}
\label{intro}

The title of this paper is a bit of a pun.  On the one hand, we are
concerned with certain natural questions about {\it automatic sequences}:
sequences over a finite alphabet where the $n$'th term is expressible
as a finite-state function of the base-$k$ representation of $n$.
On the other hand, we are interested in
answering these questions purely mechanically, in an {\it automated}
fashion.

Let ${\bf x} = (a(n))_{n \geq 0}$ be an infinite sequence over a finite
alphabet $\Delta$.  
Then $\bf x$ is said to be {\it $k$-automatic} if 
there is a deterministic finite automaton $M$ taking as input the
base-$k$ representation of $n$, and having $a(n)$ as the output
associated with the last state encountered
\cite{Allouche&Shallit:2003b}.  In this case, we say that $M$ {\it
generates} the sequence $\bf x$.

For example, in Figure~\ref{fig1}, we give an automaton generating the
well-known Thue-Morse sequence ${\bf t} = t(0) t(1) t(2) \cdots = {\tt
011010011001} \cdots$ \cite{Allouche&Shallit:1999}.  The input is $n$,
expressed in base $2$, and the output is the number contained in the
state last reached.  Thus $t(n)$ is the sum, modulo $2$, of the binary
digits of $n$.

\begin{figure}[H]
\leavevmode
\def\epsfsize#1#2{1.0#1}
\centerline{\epsfbox{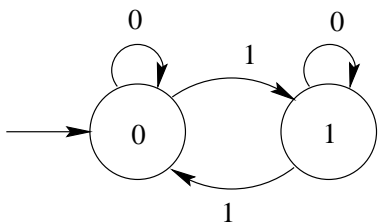}}
\protect\label{fig1}
\caption{A finite automaton generating the Thue-Morse sequence}
\end{figure}

    For at least 25 years, researchers have been interested in 
the algorithmic decidability of assertions about automatic sequences.
For example, in one of the earliest results, Honkala
\cite{Honkala:1986} showed that, given an automaton, it is decidable if
the sequence it generates is ultimately periodic.

Recently, Allouche et al.\ \cite{Allouche&Rampersad&Shallit:2009} found
a different proof of Honkala's result using a more general technique.
Using this technique, they were able to give algorithmic solutions to
many classical problems from combinatorics on words such as

\medskip

Given an automaton, is the generated sequence squarefree?  Or overlapfree?

\medskip

We write ${\bf x}[i] = a(i)$, and we let
${\bf x}[i..i+n-1]$ denote the {\it factor\/}
of length $n$ beginning at position $i$ in $\bf x$.
A sequence is said to be {\it squarefree} if it contains no factor
of the form $xx$, where $x$ is a nonempty word, and is said to
{\it overlapfree} if it contains no factor of the form $ayaya$, where $a$
is a single letter and $y$ is a possibly empty word.

The technique of Allouche et al.\ is at its core, very similar to
work of B\"uchi, Bruy\`ere, Michaux,
Villemaire, and others, involving formal logic; see, e.g.,
\cite{Bruyere&Hansel&Michaux&Villemaire:1994}.   The basic idea is as
follows:  given the automaton $M$, and some predicate $P(n)$ we want to
check, we alter $M$ by a series of transformations to a new automaton
$M'$ that accepts the base-$k$ representations of those integers $n$ for
which $P(n)$ is true.  Then we can check the assertion ``$\exists \, n \ 
P(n)$'' simply by checking if $M'$ accepts anything (which can be done
by a standard depth-first search on the underlying directed graph of
the automaton).  We can check the assertion ``$\forall \, n \ P(n)$'' by
checking if $M'$ accepts everything.  And we can check assertions like
``$P(n)$ holds for infinitely many $n$'' by checking if $M'$ has a
reachable cycle from which a final state is reachable.

Using this idea, Allouche et al.\ were able to show to reprove, purely
mechanically using a computer program, the classic theorem of Thue
\cite{Thue:1906,Thue:1912,Berstel:1995} that the Thue-Morse sequence
$\bf t$ is overlapfree.

More recently, the technique has been applied to give decision
procedures for other properties of automatic sequences.  
For example, Charlier et al.\ \cite{Charlier&Rampersad&Shallit:2011}
showed that it can be used 
to decide if a given $k$-automatic sequence
\begin{itemize}
\item contains powers of arbitrarily large exponent;
\item is recurrent;
\item is uniformly recurrent.
\end{itemize}

A sequence is said to be {\it recurrent\/} if every factor that occurs,
occurs infinitely often.  A sequence $\bf x$ is said to be
{\it uniformly recurrent} if it is recurrent and furthermore for
each finite factor $w$ occurring in $\bf x$, there is a constant $c(w)$ such
that two consecutive occurrences of $w$ are separated by at most
$c(w)$ positions.

More recently, 
variations of the technique have been used to
\begin{itemize}
\item compute the critical exponent;
\item compute the initial critical exponent;
\item decide if a sequence is linearly recurrent;
\item compute the Diophantine exponent.
\end{itemize}
(For definitions of these terms see \cite{Shallit:2011}.)

\section{The decision procedure}

In \cite{Charlier&Rampersad&Shallit:2011} we have the following
theorem:

\begin{theorem}
If we can express a property of a $k$-automatic sequence $\bf x$ using
quantifiers, logical operations, integer variables, the operations of
addition, subtraction, indexing into $\bf x$, and comparison of
integers or elements of $\bf x$, then this property is algorithmically
decidable.
\label{logic}
\end{theorem}

Let us outline how the decision procedure works.

First, the input to the decision procedure:
an automaton $M = (Q, \Sigma_k, \Delta, \delta, q_0, \tau)$
generating the $k$-automatic sequence
$\bf x$.  Here
\begin{itemize}
\item $Q$ is a nonempty set of states;
\item $\Sigma_k := \lbrace 0, 1, \ldots,
k-1 \rbrace$;
\item $\Delta$ is the output alphabet;
\item $\delta: Q \times \Sigma \rightarrow Q$ is the transition
function;
\item $q_0$ is the initial state; and
\item $\tau:Q \rightarrow \Delta$ is the output mapping.
\end{itemize}

In this paper, we assume that the automaton takes as input
the representation of $n$ in base $k$, {\it starting with the least
significant digit}; we call this the {\it reversed representation} of
$n$ and write it as $(n)_k$.  We allow leading zeroes in the representation
(which, because of our convention, are actually trailing zeroes).
Thus, for example, $011$ and $01100$ are both acceptable representations
for $6$ in base $2$.

We might also need to encode pairs, triples, or $r$-tuples of
integers.  We handle these
by first padding the reversed representation of
the smaller integer with trailing zeroes,
and then coding the $r$-tuple as a word over $\Sigma_k^r$.
For example, the pair $(20,13)$ could be represented in base-$2$ as
$$ [0,1][0,0][1,1][0,1][1,0] ,$$
where the first components spell out $00101$ and the second components
spell out $10110$. Of course,
there are other possible representations, such as
$$ [0,1][0,0][1,1][0,1][1,0][0,0],$$
which correspond to non-canonical
representations having trailing zeroes; these are also permitted.

Rather than present a detailed proof, we illustrate the idea of the
decision procedure in the proof of the following new result:

\begin{theorem}
The following problem is algorithmically decidable:
given two $k$-automatic sequences $\bf x$ and $\bf y$,
generated by automata $M_1$ and $M_2$, respectively, 
decide if $\bf x$ is a shift of $\bf y$ (that is, 
decide if there exists a constant $c$ such that
${\bf x}[n] = {\bf y}[n+c]$ for all $n \geq 0$.
\end{theorem}

\begin{proof}
We first create an NFA $M$ that accepts
the language 
$$\lbrace (c)_k \ : \ \exists n \text{ such that }
	{\bf x}[n] \not= {\bf y}[n+c] \rbrace .$$
To do so, on input $(c)_k$, $M$ 
\begin{itemize}
\item guesses $w_1 = (n)_k$ nondeterministically
(perhaps with trailing zeroes appended),
\item simulates $M_1$ on $w_1$,
\item adds $n$ to $c$ and computes the base-$k$ representation of
$w_2 = (n+c)_k$ digit-by-digit ``on the fly'', keeping track of carries,
as necessary, and simulates $M_2$ on $w_2$, and
\item accepts if the outputs of both machine differ.
\end{itemize}

We now convert $M$ to a DFA $M'$, and change final states to non-final
(and vice versa).  Then $M'$ accepts the language
$$ \lbrace (c)_k \ : \ {\bf x}[n] = {\bf y}[n+c] \text{ for all } n \geq 0
\rbrace.$$
Thus, $\bf x$ is a shift of $\bf y$ if and only if $M'$ accepts any
word, which is easily checked through depth-first search.
\qed
\end{proof}

\begin{remark}
As we can see, the size of the automata involved depends, in an unpleasant
way, on the number
of quantifiers needed to state the logical expression characterizing
the property being checked, because existential quantifiers are implemented
through nondeterminism, and universal quantifiers are implemented 
through nondeterminism and complementation
(which is implemented in a DFA
by exchange of the role
final and non-final
states).   Thus each new quantifier could increase the current
number of states, say $n$, to $2^n$ using the subset construction.
If the original automata have at most $N$
states, it follows that
the running time is bounded by an expression of the form
$$2^{2^{\Ddots^{ 2^{p(N)}}}}$$
where $p$ is a polynomial and the number of exponents in the tower
is one less than the number of quantifiers in the logical formula
characterizing the property being checked.

This extraordinary computational complexity
raises the natural question of whether the decision procedure could
actually be implemented for anything but toy examples.  Luckily
the answer seems to be yes --- at least in some cases --- as we will see below.
\end{remark}

\section{Borders}

A word $w$ is
{\it bordered} if it begins and ends with the same word $x$ with
$0 < |x| \leq |w|/2$;
Otherwise it is {\it unbordered}.  
An example in English of a bordered word is {\tt entanglement}.  
A bordered word is also called {\it bifix} 
in the literature, and unbordered words are also
called {\it bifix-free} or {\it primary}.  

Bordered and unbordered words have been actively studied in the
literature, particularly with regard to the Ehrenfeucht-Silberger
problem; see, for example,
\cite{Ehrenfeucht&Silberger:1979,Nielsen:1973,Duval:1980,Duval:1984,Harju&Nowotka:2007,Holub:2005,Costa:2003,Holub&Nowotka:2010,Rampersad&Shallit&Wang:2011,Duval&Harju&Howotka:2008}, just to name a few.

Currie and Saari \cite{Currie&Saari:2009} studied the unbordered
factors of the Thue-Morse sequence $\bf t$.  They
proved that
if $n \not\equiv \fmod{1} {6}$, then $\bf t$
has an unbordered factor of length $n$.  (Also see
\cite[Lemma 4.10 and Problem 4.1]{Saari:2008}.)  However, this is not
a necessary condition, as 
$${\bf t}[39..69] = 
{\tt 0011010010110100110010110100101},$$
which is an unbordered factor of length $31$.
Currie and Saari
left it as an open problem to give a complete characterization
of the integers $n$ for which $\bf t$ has an unbordered factor of length $n$.

The following theorem and proof, quoted practically verbatim 
from 
\cite{Charlier&Rampersad&Shallit:2011}, shows that, more generally,
the characteristic sequence of $n$ for which a given $k$-automatic
sequence has an unbordered factor of length $n$, is itself $k$-automatic:

\begin{theorem}
Let ${\bf x} = a(0) a(1) a(2) \cdots$ be a $k$-automatic sequence.
Then the associated infinite sequence
${\bf b} = b(0) b(1) b(2) \cdots$ defined by
$$
b(n) = \begin{cases}
	1, & \text{if $\bf x$ has an unbordered factor
	of length $n$;} \\
	0, & \text{otherwise;}
	\end{cases}
$$
is $k$-automatic.
\end{theorem}

\begin{proof}
The sequence $\bf x$ has an unbordered factor of length $n$

\smallskip

iff

\smallskip

\noindent $\exists j \geq 0$ such that the factor of length $n$
beginning at position $j$ of $\bf x$ is unbordered

\smallskip

iff

\smallskip

\noindent
there exists an integer $j \geq 0$ such that for all possible lengths $l$
with $1 \leq l \leq n/2$, there is an integer $i$ with $0 \leq i < l$
such that 
the supposed border of length $l$ beginning and ending the 
factor of length $n$ beginning at position $j$ of $\bf x$
actually differs in the $i$'th position

\smallskip

iff

\smallskip

\noindent there exists an integer $j \geq 0$ such that for all integers
$l$ with $1 \leq l \leq n/2$
there exists an integer $i$ with $0 \leq i < l$ such
that ${\bf x}[j+i] \not= {\bf x}[j+n-l+i]$.

\smallskip

Now assume $\bf x$ is a $k$-automatic sequence, generated by some
finite automaton.  We show how to implement the characterization given
above with an automaton.

We first create an NFA that given
the $(j,l,n)_k$ guesses the base-$k$
representation of $i$, digit-by-digit, checks that $i < l$,
computes $j+i$ and $j+n-l+i$ on the fly, and 
checks that ${\bf x}[j+i] \not= {\bf x}[j+n-l+i]$.  If such an $i$ is found,
it accepts.  We then convert this to a DFA, and interchange accepting
and nonaccepting states.  This DFA $M_1$ accepts
$(j,l,n)_k$ such that there is no $i$, $0 \leq i < l$ such that
${\bf x}[j+i] = {\bf x}[j+n-l+i]$.   We then use $M_1$ as a subroutine to 
build an NFA $M_2$ that on input $(j,n)_k$
guesses $l$, checks that $1 \leq l \leq n/2$, and
calls $M_1$ on the result.  We convert this to a DFA and interchange
accepting and nonaccepting states to get $M_3$.  Finally, this
$M_3$ is used as a subroutine to build an NFA $M_4$ that
on input $n$ guesses $j$ and calls $M_3$.  

The characteristic sequence of these integers $n$ 
is therefore $k$-automatic.
\qed
\end{proof}

Since the proof is constructive, one can, in principle, carry out the
construction to get an explicit description of the lengths for which
the Thue-Morse sequence has an unbordered factor.

Doing so results in the following theorem:

\begin{theorem}
There is an unbordered factor of length $n$ in $\bf t$ if and only if
the base-$2$ representation of $n$ (starting with the most significant
digit) is not of the form $1 (01^*0)^* 1 0^* 1$.
\end{theorem}

\begin{proof}
The proof of this theorem is purely mechanical, and it involves
performing a sequence of operations on finite automata.  The second
author wrote a program in C++, using his own automata package,
to perform these operations.  There are 
four stages to the computation, which are described in detail below.

\medskip

\noindent{\bf Stage 1}\\

Let $T$ be the automaton of Figure~\ref{fig1} generating the
Thue-Morse sequence $\bf t$.
Stage 1 takes $T$ as input and outputs an automaton $M_1$, where $M_1$
accepts $w \in (\{0,1\}^4)^*$ if and only if $w$ is the
base-$2$ representation
of some $(n,j,l,i) \in S_1$, where
\begin{equation}
S_1 = \{ (n,j,l,i) \ : \ 0 < l \le n/2 \text{ and } i < j \text{ and } \textbf{t}[j+i] \ne \textbf{t}[n+j-l+i]  \}.
\end{equation}

The size of $M_1$ was only 102 states.
However, since the input alphabet for
$M_1$ is of size $2^4 = 16$,
a considerable amount of complexity is being stored in
the transition matrix.
Stage 1 passed all 1.3 million tests meant to ensure that
$M_1$ corresponds to $S_1$.
\\

\noindent \textbf{Stage 2}\\

The purpose of Stage 2 is to remove the variable $i$ by simulating it.
The resulting machine, after being negated, accepts $(n,j,l)$ iff the length $n$
factor of $t$ starting at index $j$ has a border of length $l$.  So
Stage 2 produces the automaton $M_2$, which is the negation of the
result of simulating $i$. More formally, $M_2$ accepts a word $w \in
(\{0,1\}^3)^*$ if and only if $w$ is the base-$2$ representation
of some $(n,j,l) \in
S_2$, where
\begin{equation}
S_2 = \{ (n,j,l) : \not\exists i \textrm{ for which } (n,j,l,i) \in S_1 \}
\end{equation}

The size of $M_2$ after subset construction was 8689 states, and it 
minimized down to 127 states. The output of Stage 2 passed all 1.6
million tests meant to ensure that $M_2$ corresponds to $S_2$.\\

\noindent \textbf{Stage 3}\\

The purpose of Stage 3 is to remove $l$ by simulating it. By the end of
Stage 3, most of the work has already been done. The output of Stage 3,
$M_3$, accepts an input word $w \in (\{0,1\}^2)^*$ if and only if $w$
is the base-$2$ representation of some $(n,j) \in S_3$, where
\begin{equation}
S_3 = \{ (n,j) : \not\exists l \textrm{ such that } (n,j,l) \in S_2 \}
\end{equation}
or, in other words
\begin{equation}
S_3 = \{ (n,j) : \textrm{ \textbf{t} has an unbordered factor of length $n$ at index $j$} \}.
\end{equation}

The size of $M_3$ after subset construction was 1987 states, and it
minimized down to 263 states. The output of Stage 3 passed all 1.9
million tests meant to ensure that $M_3$ corresponds to $S_3$.\\

\noindent \textbf{Stage 4}\\

Finally, Stage 4 simulates $j$ on $M_3$ and negates the result. So the
output of Stage 3 is an automaton that accepts the binary
representation of a positive integer $n>1$ if and only if the Thue-Morse
word has no unbordered factor of length $n$. Formally put, the
automaton $M_4$ produced by Stage 4 accepts a word $w \in \{0,1\}^*$ if
and only if $w$ is the base-$2$ representation of some $n \in S_4$, where
\begin{equation}
S_4 = \{ n \in \mathbb{N} \ : \ n > 1, \not\exists j \textrm{ for which } (n,j) \in S_3\}.
\end{equation}

The size of $M_4$ after subset construction is 2734 states, and it minimized to
7 states.  $M_4$ accepts the reverse of $1(01^*0)^*10^*1$. Therefore
the Thue-Morse word has an unbordered factor of length $n$ if and only
if the base-2 representation of $n$ (starting with the most significant
digit) is not of the form $1(01^*0)^*10^*1$.

The total computation took 9 seconds of CPU time on a 2.9GHz Dell XPS laptop.
\qed
\end{proof}

\begin{remark}
Here are some additional implementation details.
	
In order to implement the needed operations on automata, we
must decide on an encoding of elements of $(\Sigma_k^n)^*$.
We could do this by performing a perfect shuffle of each individual
word over $\Sigma_k^*$, or by letting the alphabet itself be represented
by $k$-tuples.  The decision represents a tradeoff between state size
and alphabet size.  We used the latter representation, since
(a) it makes the algorithms considerably easier to
implement and understand and (b)
decreases the number of states needed.

It was mentioned earlier how many tests were passed in each
stage. In order to make sure that the final automaton is what
we expect, a number of tests are run after each stage
on the output of that stage.

For example, let $\textbf{x}$ be an automatic sequence. The
testing framework requires a C++ function which given $n$
computes $\textbf{x}[n]$. Before any operations are
done, the automaton given for $\textbf{x}$ is tested against
the C++ function to make sure that they match for the first
10,000 elements. Then, at each stage before Stage 4 the
resulting automaton is tested to give confidence that the
operations on the automata are giving the desired results.

For example, after Stage 2 of computing the set of lengths for
which there exists an unbordered factor of an automatic
sequence $\textbf{x}$, we expect the machine $M_2$ to accept
the language $S_2$, where \begin{equation} S_2 = \{(n,j,l) :
\not\exists i \textrm{ for which }
\textbf{x}[j+i]=\textbf{x}[n+j-l+i]\} \end{equation} This is
then tested by making sure $M_2$ accepts $(n,j,l)_k$ if and
only if $(n,j,l) \in S_2$ for all $n,j,l \leq 1400$.
These tests were invaluable to 
debugging, and provide confidence in the final result of
the computation.

Finally, we have to address the issue of multiple representations.  It
is easy to forget that automata accept words in ${\Sigma_k}^*$, and not
integers. 
For some operations, such as
complement and intersection, it is crucial that if one binary
representation is accepted by the automaton, then all binary
representations must be accepted.
\end{remark}

\section{Additional results}

We also applied our decision procedure above to two other famous sequences:
the Rudin-Shapiro sequence \cite{Rudin:1959,Shapiro:1952}
and the paperfolding sequence \cite{Dekking&MendesFrance&vanderPoorten:1982}.

For a word $w \in 1 (0+1)^*$, we define $a_w (n)$ to be the number
of (possibly overlapping) occurrences of $w$ in the (ordinary,
unreversed) base-$2$ representation of $n$.  Thus, for example,
$a_{11} (7) = 2$.

The {\it Rudin-Shapiro} sequence ${\bf r} = r(0) r(1) r(2) \cdots$
is then defined to be $r(n) = (-1)^{a_{11} (n)}$.  It is a $2$-automatic
sequence generated by an automaton of four states.

The {\it paperfolding sequence} ${\bf p} = p(0) p(1) p(2) \cdots$
is defined as follows:  writing 
$(n)_2 00$ as $1^i 0 aw$ for some $i \geq 0$ some
$a \in \lbrace 0, 1\rbrace$, and some $w\in \lbrace 0,1 \rbrace^*$, 
we have $p(n) = (-1)^a$.  It is a $2$-automatic sequence generated by
an automaton of four states.

\begin{theorem}
The Rudin-Shapiro sequence has an unbordered factor of every length.
\end{theorem}

\begin{proof}
We applied the same technique discussed previously for
the Thue-Morse sequence.

Here is a summary of the computation:

\noindent Stage 1:  269 states\\
Stage 2:  85313 states minimized to 1974\\
Stage 3:  48488 states minimized to 6465\\
Stage 4:  6234 states. \\

The Stage 4 NFA has 6234 states.  We were unable to determinize this
automaton directly (using two different programs) due to an explosion
in the number of states created.  Instead, we reversed the NFA 
(creating an NFA for $L^R$) and determinized this instead.  The
resulting DFA has 30 states, and upon minimization, gives a $1$-state
automaton accepting all strings.
\qed
\end{proof}

\begin{theorem}
The paperfolding sequence has an unbordered factor of length
$n$ if and only if the reversed representation $(n)_2$ is
rejected by the automaton given in Figure~\ref{fig2}.
\end{theorem}

\begin{figure}[htb]
\leavevmode
\def\epsfsize#1#2{0.5#1}
\centerline{\epsfbox{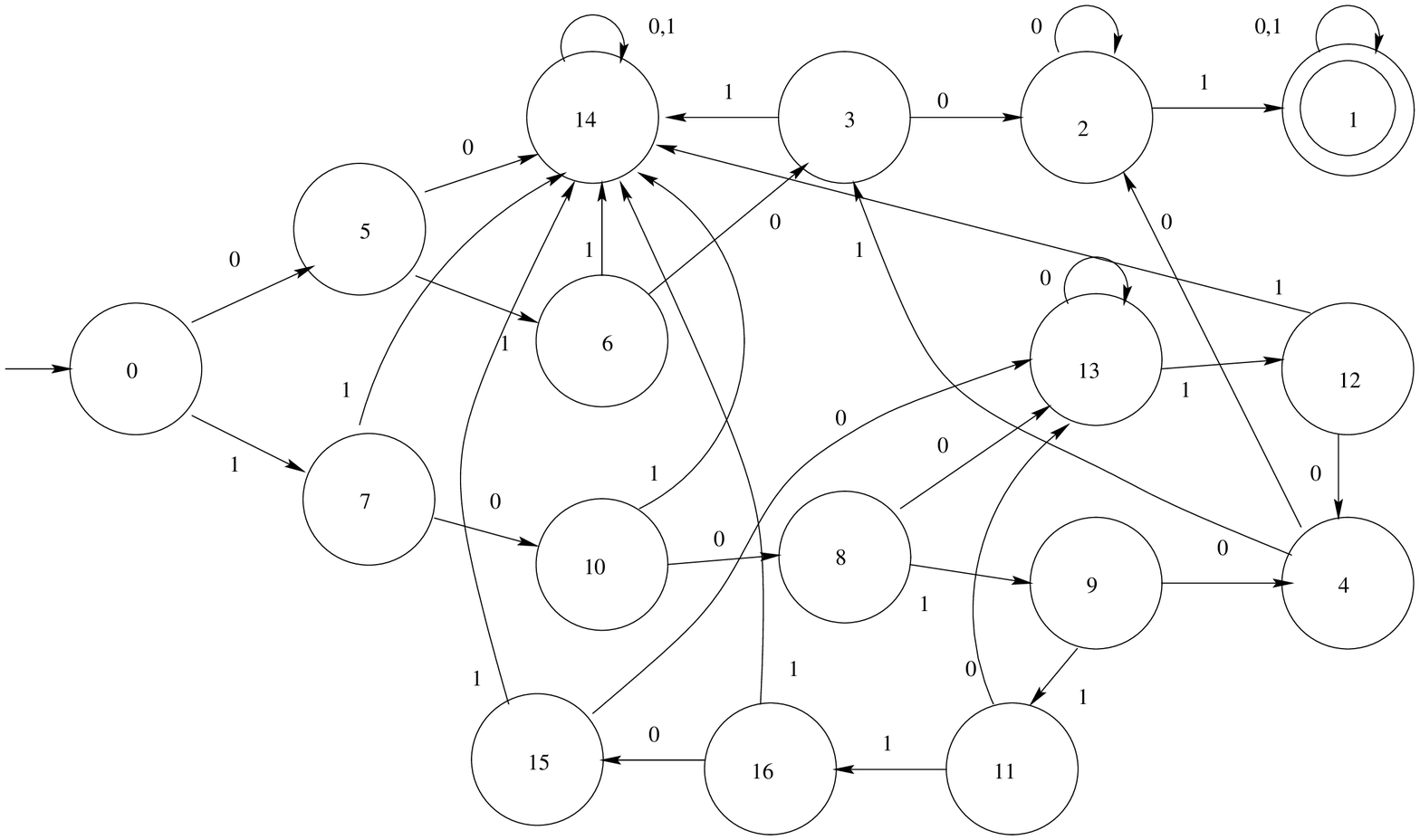}}
\protect\label{fig2}
\caption{A finite automaton for unbordered factors in the paperfolding word}
\end{figure}

\begin{proof}
We applied the same technique discussed previously for the Thue-Morse
sequence.

Here is a summary of the computation:  6 seconds cpu time on a 2.9GHz Dell XPS laptop.

\medskip

\noindent Stage 1, 159 states \\
Stage 2, 1751 minimized down to 89 states \\
Stage 3, 178 minimized down to 75 states \\
Stage 4, 132 minimize down to 17 states .
\qed
\end{proof}

\section{Further work}

In the future, we plan to extend this work to explicitly compute the
number of distinct unbordered factors of length $n$ in the Thue-Morse sequence.
(A conjecture about this number was given in
\cite{Charlier&Rampersad&Shallit:2011}.)

\section{Open problems}

Which of the problems
mentioned in \S~\ref{intro} are algorithmically decidable for the more
general class of morphic sequences?

Can the techniques be applied to detect
abelian powers in automatic sequences?

\end{document}